\documentclass[10pt,aps,prl,nofootinbib,longbibliography,notitlepage,twocolumn,floatfix]{revtex4-2}

\usepackage{amsthm}
\usepackage{paralist}
\usepackage{tikz}
\usepackage{graphicx}
\usepackage{tabularx}
\usepackage{amsmath}
\usepackage{amstext}
\usepackage{amssymb}
\usepackage{amsfonts}
\usepackage{amsbsy}
\usepackage{xfrac,xcolor}
\usepackage{bbm}
\usepackage{relsize}
\usepackage[colorlinks,citecolor=blue]{hyperref}
\usepackage{wrapfig}
\usepackage{longtable}
\usepackage{booktabs}
\usepackage{url}
\usepackage{subfigure}
\usepackage{dsfont}
\usepackage{dcolumn}
\usepackage[inline]{enumitem}
\usepackage{bm}
\usepackage{esint}
\usepackage{multirow}
\usepackage{cleveref}
\usepackage{mathrsfs}
\usepackage{physics}
\usepackage{floatflt}
\usepackage{extarrows}
\usepackage{orcidlink}
\usepackage{datetime}
\usepackage{comment}
\usepackage[super]{nth}
\usepackage{tikz-cd}
\usepackage{sidecap}
\usepackage{outlines}
\usepackage{lineno}

\usetikzlibrary{arrows, shapes.misc, shadows.blur}

\setlist[enumerate]{left=0pt, labelwidth=!,parsep=0pt,topsep=0pt,itemsep=0pt}
\newcommand{\id}[1]{\left\langle{#1}\right\rangle}
\newcommand{\ff}{\mathbb{F}_2[x,y,x^{-1},y^{-1}]}
\newcommand{\conv}{\operatorname{Conv}}
\newcommand{\supp}{\operatorname{Supp}}
\newcommand{\np}{\operatorname{Newt}}
\def\m{\mathfrak{m}}

\newcommand{\mb}{\mathbb}

\newcommand{\mc}{\mathcal}

\def\E{\mathcal{E}}

\definecolor{fractoncolor}{HTML}{332237}
\definecolor{lineoncolor}{HTML}{918CC6}
\definecolor{mobilecolor}{HTML}{EBC4E3}
\definecolor{emptycolor}{RGB}{245, 245, 248}

\newtheoremstyle{mystyle}
{\topsep}
{\topsep}
{\normalfont}
{}
{\bfseries}
{.}
{ }
{}

\newtheorem{theorem}{Theorem}
\newtheorem{lemma}{Lemma}
\newtheorem{definition}{Definition}

\allowdisplaybreaks[4]

\tikzset{greensq/.style={rectangle,fill=green,draw=black,line width=1.5pt}}
\tikzset{grayop/.style={circle,fill=gray,draw=black,line width=1.5pt}}
\tikzset{bold/.style={color=blue, line width=4pt}}
\tikzset{redop/.style={circle,fill=red,draw=black,line width=1.5pt}}
\tikzset{blueop/.style={circle,fill=blue,draw=black,line width=1.5pt}}
\tikzset{greenop/.style={circle,fill=green,draw=black,line width=1.5pt}}

\hypersetup{
	colorlinks=true,
	linkcolor=blue,
	filecolor=magenta,
	urlcolor=magenta,
}

\renewcommand{\theparagraph}{\arabic{paragraph}.}
\usepackage{titlesec}
\titleformat{\section}[runin]
  {\itshape\normalsize\bfseries}
  {\arabic{section}.}
  {0.5em}
  {}
\titlespacing*{\section}{10pt}{1ex}{0.5ex}
\renewcommand{\thesection}{\arabic{section}}

\titleformat{\subsection}[runin]
  {\itshape\normalsize}
  {\arabic{section}.}
  {0.5em}
  {}
\titlespacing*{\subsection}{0pt}{0.5ex}{0.5ex}
\renewcommand{\thesubsection}{\arabic{section}}

\begin{document}

\begin{abstract}
    In topological phases of matter, fusion rules dictate how anyonic topological charges combine. However, the transformation of quasiparticle mobility under fusion remains largely unexplored. In this Letter, we reveal that restricted mobility classes obey their own complex multi-channel fusion rules. We introduce three exactly solvable models with $\mathbb{Z}_2$ topological order enriched by subsystem symmetries to explicitly demonstrate these structures. Within these models, mobility constraints arise from enforcing symmetries supported on specific subsets. When excitations fuse, these rigid geometric constraints interfere spatially. At the macroscopic level, this deterministic geometric interference manifests as multi-channel fusion rules. We present three explicit examples: (i) Fibonacci fusion rules; (ii) tensor products of Fibonacci rules; and (iii) lineon period transmutation.
\end{abstract}

\title{Fusion Rules of Mobility}
\author{Jie-Yu Zhang\orcidlink{0000-0003-4414-7680}}
\author{Peng Ye\orcidlink{0000-0002-6251-677X}}
\email{yepeng5@mail.sysu.edu.cn}
\affiliation{
    Guangdong Provincial Key Laboratory of Magnetoelectric Physics and Devices, State Key Laboratory of Optoelectronic Materials and Technologies, and
    School of Physics, Sun Yat-sen University, Guangzhou, 510275, China}
\date[]{Published in Phys. Rev. B \textbf{114}, L171107 (2026), DOI: \href{https://doi.org/10.1103/jrdw-vvct}{10.1103/jrdw-vvct}}
\maketitle

\textit{Introduction.---} Anyon fusion is a central concept in topological phases of matter~\cite{Tsui1982,Laughlin1983,Savary2017,Broholm2020}, establishing the algebraic rules for combining topological charges. While these charge fusion rules are well understood, the topological charge is not the only property altered during fusion. A fundamental yet underexplored aspect is the transformation of quasiparticle mobility.

In exotic phases of matter, quasiparticle dynamics can be highly restricted, dividing excitations into distinct mobility classes such as fully mobile anyons, sub-dimensional lineons, and immobile fractons~\cite{chamon2005,Haah2011local,haah2015,fracton1,Nandkishore2019,Pretko2020review,LiYe2020I,LiYe2021II,BoxiLi2024Toeplitz}. While the fusion of their topological charges may be trivial, the fusion of their \textit{mobility classes} exhibits novel behaviors. For instance, in the X-cube model~\cite{fracton1}, two fractons can fuse into a lineon, and two lineons can fuse into a planon. Traditionally, such mobility transformations have been viewed as informal, kinematic consequences of spatial constraints. This naturally raises a fundamental question: does an underlying, strict algebraic structure—a ``fusion rule''—govern mobility classes themselves?

In this Letter, we answer affirmatively: mobility classes are governed by complex multi-channel fusion rules. We demonstrate this structure within subsystem symmetry-enriched topological (SSET) phases~\cite{Stephen2022,SSET1,Barkeshli2019,babakhani2023g,Ning2016SET,Ning2022Loop,Zhou2022SSPT}. In systems with subsystem symmetry—a subclass of generalized symmetry~\cite{You2018,gaiotto_generalized_2015,yoshida_topological_2016,mcgreevy_generalized_2023,shao_whats_2024,tachikawa_gauging_2020,ji_categorical_2020,wen_emergent_2019}—mobility constraints are enforced by symmetries whose support lies on specific subsets of the system~\cite{you2019}. Consequently, each mobility class carries a rigid geometric symmetry constraint.

When two excitations fuse, their respective subsets interfere spatially. The resulting mobility of the fused excitation is generally not unique; it is highly sensitive to the microscopic relative position of the fusing constituents. In exact lattice models of non-Abelian topological orders (e.g., quantum double models), macroscopic multi-channel fusion outcomes necessitate hidden microscopic degrees of freedom~\cite{buerschaper2009mapping,padmanabhan2015nonabelian,huang2025bridging,Zhang2023NAfusion,li2026gapped,Huang2023JHEP5D}. Here, the precise microscopic relative position acts as this hidden degree of freedom. At the macroscopic level, where this exact spatial separation is coarse-grained, the deterministic geometric interference manifests as multi-channel mobility fusion rules.

To concretely demonstrate this algebraic structure, we construct three exactly solvable models featuring $\mathbb{Z}_2$ topological order enriched by tunable subsystem symmetries. By varying the model construction, we alter the subsystem symmetry constraints, yielding explicit mobility fusion rules. We detail three distinct scenarios: Fibonacci-like fusion rules, tensor products of Fibonacci rules, and lineon period transmutation. Crucially, this multi-channel nature characterizes the mobility fusion rules and is distinct from non-Abelian braiding statistics.

\textit{Exactly Solvable Models.---} Specifically, we consider qubit models defined on an infinite 2D square lattice, with qubits residing on both vertices and edges (3 qubits per unit cell). The Hamiltonian consists of three mutually commuting stabilizer terms, defined as follows:
\begin{equation}
    H = -\sum_{v} A_v - \sum_{p} B_p - \sum_{v} C_v.
\end{equation}
The term $A_v$ is a decorated version of the toric code vertex term. For a vertex $v$ at position $(i_0,j_0)$ (hereafter, $i,j,i_0,j_0$ take values in $\mathbb{Z}$), the vertex term reads:
\begin{equation}
    A_v = \left[ \prod_{e \in \mathrm{star}(v)} X_e \right] \otimes \left[ \prod_{v' \in S_v(f)} Z_{v'} \right],
\end{equation}
where $\mathrm{star}(v)$ denotes the four edges adjacent to $v$, and $S_v(f)$ is the set of vertices where $Z$ operators are placed according to the decoration pattern. This pattern is conveniently encoded by a Laurent polynomial $f(x,y) \in \mathbb{F}_2[x,y,x^{-1},y^{-1}]$:
\begin{equation}
    f(x,y) = \sum_{i,j} c_{ij} x^i y^j, \quad c_{ij}\in \{0,1\}.
\end{equation}
Each monomial $x^i y^j$ in $f$ with $c_{ij}=1$ corresponds to a $Z$ operator at $(i_0+i, j_0+j)$. Different choices of $f(x,y)$ yield distinct subsystem symmetry structures while keeping the underlying topological order invariant.

For a plaquette $p$ centered at $(i_0+0.5, j_0+0.5)$, the $B_p$ term is
\begin{equation}
    B_p = Z_{(i_0+0.5, j_0)} Z_{(i_0+0.5, j_0+1)} Z_{(i_0, j_0+0.5)} Z_{(i_0+1, j_0+0.5)}.
\end{equation}

The final term, $C_v$, geometrically couples the toric code flux ($m$-anyons) to the subsystem symmetries dictated by $f(x,y)$. It is defined as
\begin{equation}
    C_v = X_{(i_0,j_0)} \otimes \left[ \prod_{e \in T_v(f)} Z_e \right].
\end{equation}
The spatial configuration of $T_v(f)$ is chosen such that $C_v$ commutes with all other Hamiltonian terms. For the three representative models studied here, the corresponding $T_v(f)$ are constructed and verified algebraically in the Supplemental Material~\cite{sm}; the patterns are illustrated in the third column of Fig.~\ref{fig:Hamiltonian}.

In Fig.~\ref{fig:Hamiltonian} we illustrate the three representative cases: \textit{Model (a)} ($f_1=1+y$, Fig.~\ref{fig:Hamiltonian}(a)), \textit{Model (b)} ($f_2=(1+x)(1+y)$, Fig.~\ref{fig:Hamiltonian}(b)), and \textit{Model (c)} ($f_3=(1+x)(1+y^2)(1+y+y^2)$, Fig.~\ref{fig:Hamiltonian}(c)).

\begin{figure}[ht]
    \centering
    \includegraphics{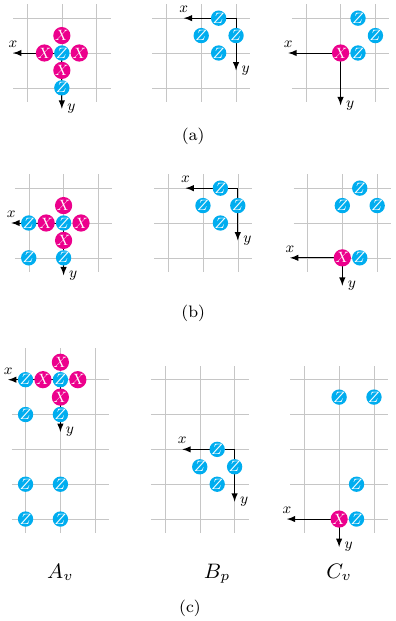}
    \caption{Lattice stabilizers for three representative models, where Pauli $Z$ and $X$ operators are denoted by blue and magenta circles, respectively. Hamiltonian terms of \textit{model (a), (b), (c)} are respectively shown in Fig.~\ref{fig:Hamiltonian}(a),~\ref{fig:Hamiltonian}(b),~\ref{fig:Hamiltonian}(c). Each row shows the three stabilizer types for a specific choice of the decoration pattern $f(x,y)$.}
    \label{fig:Hamiltonian}
\end{figure}

\begin{figure}[htbp]
    \centering
    \includegraphics{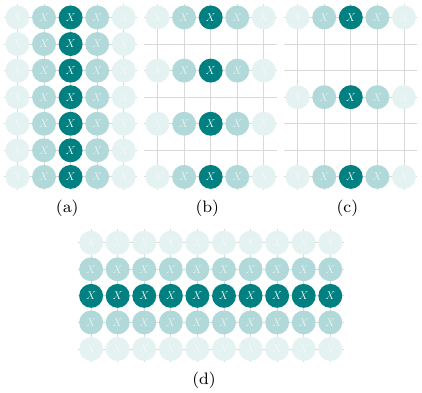}
    \caption{Representative subsystem symmetry generators for the three models. Pauli-$X$ operators on vertices are shown with a finite segment for clarity, but extend infinitely. Translations by one lattice spacing—horizontal (a--c) or vertical (d)—also give valid generators.}
    \label{fig:symmetry_patterns}
\end{figure}

On a torus, the model exhibits $\mathbb{Z}_2$ topological order with a ground-state degeneracy of 4. The models feature standard toric code topological excitations: $e$-type (violating $A_v$) and $m$-type (violating $B_p$). From the perspective of symmetries, these models possess additional subsystem symmetries whose generators $S$ consist of Pauli-$X$ operators on the vertices. A typical symmetry generator is $S(f)=\prod_{v\in \supp(\text{Symmetry})}X_v$, where $\supp(\text{Symmetry})$ denotes the support of a given subsystem symmetry generator (e.g., all vertices on a specific line or a fractal subset of the lattice).
The decoration polynomial $f(x,y)$ fixes a local constraint on the symmetry supports. This relation can also be formulated in the higher-order-cellular-automaton language of Refs.~\cite{PRXQuantum.5.030342,Huang2026NUCA,Stephen2019QCA}---a framework also widely used in classical computation for secret sharing~\cite{MARTINDELREY20051356}, encryption~\cite{CA1,DENNUNZIO2024119942}, and data compression~\cite{CA2}; thus, the spatial patterns of the subsystem symmetries can be tuned by choosing different $f(x,y)$. For the three examples introduced above, their subsystem symmetry supports are illustrated in Fig.~\ref{fig:symmetry_patterns}. \textit{Model (a)} (decorated by $f_1$) possesses subsystem symmetries along all vertical lines (Fig.~\ref{fig:symmetry_patterns}(a)), while \textit{Model (b)} (decorated by $f_2$) has subsystem symmetries along both horizontal (Fig.~\ref{fig:symmetry_patterns}(d)) and vertical (Fig.~\ref{fig:symmetry_patterns}(a)) lines. \textit{Model (c)} (decorated by $f_3$) has subsystem symmetries along all horizontal lines (Fig.~\ref{fig:symmetry_patterns}(d)) but two distinct types of vertical symmetry generators: one with Pauli-$X$ operators every two lattice spacings (Fig.~\ref{fig:symmetry_patterns}(b)), and another every three lattice spacings (Fig.~\ref{fig:symmetry_patterns}(c)).

These three examples demonstrate how interference among distinct subsystem-symmetry constraints produces nontrivial mobility fusion rules.

\textit{Mobility Classes and Mobility Fusion Rules.---}
Under the constraint of preserving all subsystem symmetries, anyon mobility in these models becomes restricted. The underlying mechanism is that each excitation must preserve its charge under all subsystem symmetry generators. In the three models discussed above, the $C_v$ terms effectively impose symmetry constraints on $m$-anyons. Moving an $m$-anyon requires Pauli-$X$ string operators on the edges, which may create additional $C_v$ excitations. The local generation and annihilation of a $C_v$ excitation involves a single Pauli-$Z$ operator on a vertex, which anticommutes with a set of subsystem symmetry generators. Thus, the mobility of an $m$-anyon is restricted by the condition that the string operator moving it must commute with every subsystem symmetry generator. In these models, $e$-anyons remain fully mobile under symmetry restrictions.

By uniquely assigning a vertex $(i,j)$ to each $m$-anyon at plaquette $(i+0.5,j+0.5)$, the condition above requires, and is equivalent to, the number of $m$-anyons being conserved (mod 2) within every $\supp(\textrm{Symmetry})$.

As a simple example, in \textit{Model (a)}, where subsystem symmetries are vertical line-like, a single $m$-anyon can only move vertically while respecting all subsystem symmetries. However, two vertically adjacent $m$-anyons---forming a composite excitation after fusion---are fully mobile, since the number of $m$-anyons in every vertical line remains $0 \pmod 2$. For any configuration of $m$-anyons, represented by $\mathfrak{m}(x,y)=\sum_{i,j}m_{ij}x^iy^j \in \mathbb{F}_2[x,y,x^{-1},y^{-1}]$, where $m_{ij}=1$ denotes an $m$-anyon at plaquette $(i+0.5,j+0.5)$, we define its mobility class $\mathcal{M}(\mathfrak{m})$ as the set of positions to which the configuration can be moved by a symmetry-respecting local operator. In these models,
$\mathcal{M}(\mathfrak{m})\in\{\alpha,\beta_{\mathbf v,T},\gamma\}.$
Here, the fully mobile class $\alpha$ moves freely in all directions;
lineons ($\beta_{\mathbf{v},T}$) move only along one-dimensional subsystems in direction $\mathbf{v}=(a,b),~\gcd (a,b)=1$ with a minimal hopping period $T$, meaning that the minimal displacement is $(\pm aT,\pm bT)$;
and fractons ($\gamma$) are immobile without breaking the subsystem symmetries.

The example above explicitly shows that anyon fusion can significantly alter the mobility class due to the spatial interference of constraints carried by each individual anyon. Given two distinct excitations, the mobility of their fusion outcome is therefore naturally dependent on their relative positions. One may wonder whether there is a hidden algebraic structure beneath this mobility fusion process; the answer is affirmative.

We demonstrate that the \textit{mobility fusion rules} arise from the following process: Each excitation configuration carries a distinct set of symmetry constraints depending on the positions of the constituent anyons. When two excitations fuse, their combined constraints depend sensitively on their relative positions: at some separations the constraints partially cancel, opening up new mobility channels; at others they reinforce, further restricting mobility.

We define mobility fusion rules at both the mobility-class and excitation-configuration levels. For mobility classes $\mathcal{M}_1$ and $\mathcal{M}_2$, the rule is
\begin{equation}\label{eq:mf_class}
    \mathcal{M}_1\times \mathcal{M}_2=\sum_{\mathcal{M}_{\lambda}\in\text{fusion channels}}\mathcal{M}_{\lambda},
\end{equation}
where $\mathcal{M}_\lambda$ is a valid fusion channel if and only if there exist $\mathfrak{m}_1, \mathfrak{m}_2$ satisfying $\mathcal{M}(\mathfrak{m}_1)=\mathcal{M}_1$ and $\mathcal{M}(\mathfrak{m}_2)=\mathcal{M}_2$, such that $\mathcal{M}(\mathfrak{m}_1+\mathfrak{m}_2)=\mathcal{M}_\lambda$.

For excitation configurations $\mathfrak{m}_1$ and $\mathfrak{m}_2$, the mobility fusion rule is
$
    \mathfrak{m}_1\times \mathfrak{m}_2 =\sum_{\mathcal{M}_{\lambda}\in\text{fusion channels}}\mathcal{M}_{\lambda}$,
where $\mathcal{M}_\lambda$ is summed over all channels such that there exists some displacement $x^i y^j$ (equivalently, they are placed at relative position $(i,j)$) with $\mathcal{M}(\mathfrak{m}_1+x^iy^j \mathfrak{m}_2)=\mathcal{M}_\lambda$.
Here and below, the sum and the symbol $+$ are shorthand for the set of realizable outcomes, not quantum fusion multiplicities. We use ``mobility fusion rule'' for this set-valued composition of mobility data, in parallel with the conventional fusion rule for topological charges: the latter composes charge sectors, whereas the former composes mobility classes. The labels retain only mobility information and do not distinguish topological charge; in particular, $\alpha$ includes locally trivial configurations, which are mobile in this sense.
Although the mobility class $\mathcal{M}(\mathfrak{m})$ of a specific excitation depends strongly on the decoration pattern $f(x,y)$, the three examples below share a common factorization-based structure.

\begin{figure*}[t]
    \centering
    \includegraphics{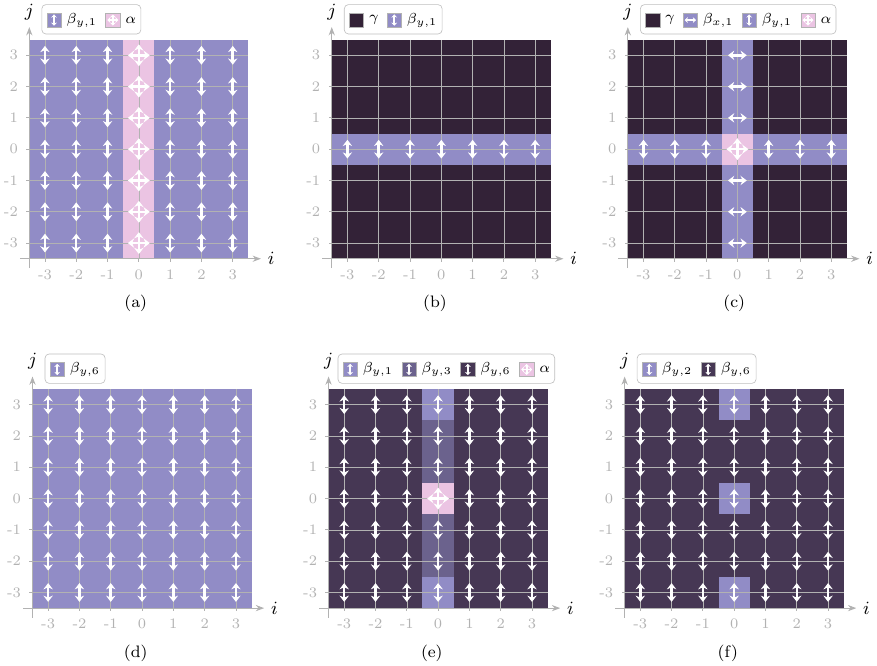}
    \caption{Mobility-fusion maps for two $m$-anyon configurations as functions of their relative displacement $(i,j)$. Pink cells with four arrows denote the fully mobile class $\alpha$; purple cells with a double-headed arrow denote a lineon class $\beta$; dark-gray cells denote the immobile class $\gamma$. The lineon period is given in each panel's legend. (a) Model (a), with $\mathfrak{m}_1=\mathfrak{m}_2=1$. (b)--(c) Model (b): $\mathfrak{m}_1=1+y,~\mathfrak{m}_2=1+xy$ in (b), and $\mathfrak{m}_1=1+xy,~\mathfrak{m}_2=x+y$ in (c). (d)--(f) Model (c): $\mathfrak{m}_1=(1+x)(1+y+y^2),~\mathfrak{m}_2=(1+x)(1+y^2)$ in (d); $\mathfrak{m}_1=\mathfrak{m}_2=1+x$ in (e); and $\mathfrak{m}_1=1+x,~\mathfrak{m}_2=(1+x)(1+x+xy+xy^2)$ in (f).}
    \label{fig:mobility_map}
\end{figure*}

We now explicitly present the mobility fusion rules for \textit{Models (a)}--\textit{(c)}. In each example, the class-level mobility fusion rule is computed algebraically and illustrated by one or two excitation-configuration pairs.

    \textbf{\textit{Model Example (a): Fibonacci Fusion Rules---}}
For \textit{Model (a)}, decorated by $f_1(x,y) = 1+y$, there are a total of two mobility classes: fully mobile anyons $\alpha$ and lineons $\beta_{y,1}$, which move along the $y$-axis with a period of 1. Remarkably, these mobility classes for Abelian anyons respect the fusion rules whose mathematical expressions are exactly the same as the standard fusion rules for Fibonacci anyons (i.e., $\alpha\leftrightarrow \mathbb{I}$ and $\beta_{y,1}\leftrightarrow \tau$):
\begin{subequations}
    \begin{align}
        \alpha\times \alpha            & =\alpha                                 \\
        \alpha\times\beta_{y,1}        & =\beta_{y,1}                            \\
        \beta_{y,1} \times \beta_{y,1} & = \alpha + \beta_{y,1}. \label{eq:fib1}
    \end{align}
\end{subequations}
The multi-channel fusion structure arises from the position-dependent interference of symmetry constraints. While it is impossible to show every fusion process between two mobility classes, we can explicitly demonstrate these fusion rules by selecting two representative excitations, varying their relative position $(i,j)$, and calculating the mobility of the resulting composite. The result is pictorially represented by a \textit{mobility map}, from which the mobility $\mathcal{M}(\mathfrak{m}_1+x^iy^j \mathfrak{m}_2)$ can be read at position $(i,j)$, where $\mathfrak{m}_2$ is placed at relative position $(i,j)$ from $\mathfrak{m}_1$. The mobility map illustrating Eq.~\eqref{eq:fib1} with $\mathfrak{m}_1=\mathfrak{m}_2=1$ is shown in Fig.~\ref{fig:mobility_map}(a).

\textbf{\textit{Model Example (b): Tensor Product of Fibonacci Fusion Rules---}} For \textit{Model (b)}, decorated by $f_2(x,y) = (1+x)(1+y)$, there are four mobility classes: fully mobile $\alpha$, horizontal lineons $\beta_{x,1}$, vertical lineons $\beta_{y,1}$, and fractons $\gamma$. The mobility fusion rules are:
\begin{subequations}
    \begin{align}
         & \beta_{x,1} \times \beta_{x,1} = \alpha + \beta_{x,1}, \quad \beta_{y,1} \times \beta_{y,1} = \alpha + \beta_{y,1}, \\
         & \beta_{x,1} \times \beta_{y,1} = \gamma,                                                                            \\
         & \beta_{x,1} \times \gamma = \beta_{y,1} + \gamma, \quad \beta_{y,1} \times \gamma = \beta_{x,1} + \gamma,           \\
         & \gamma \times \gamma = \alpha + \beta_{x,1} + \beta_{y,1} + \gamma.
    \end{align}
\end{subequations}

A horizontal lineon and a vertical lineon, when fused, inherit the mobility constraints of both directions---and since they cannot simultaneously satisfy the constraints needed for motion in either direction, the composite becomes immobile (a fracton).

These mobility fusion rules have the product-form binary expressions obtained from two Fibonacci rules, under the identification $\alpha\leftrightarrow(1,1)$, $\beta_{x,1}\leftrightarrow(1,\tau)$, $\beta_{y,1}\leftrightarrow (\tau,1)$, and $\gamma \leftrightarrow(\tau,\tau)$.

Two typical mobility maps for this model are presented in Fig.~\ref{fig:mobility_map}(b) ($\mathfrak{m}_1=1+y,~\mathfrak{m}_2=1+xy$) and Fig.~\ref{fig:mobility_map}(c) ($\mathfrak{m}_1=1+xy,~\mathfrak{m}_2=x+y$).

\textbf{\textit{Model Example (c): Lineon Period Transmutation---}} For \textit{Model (c)}, decorated by $f_3(x,y) = (1+x)(1+y^2)(1+y+y^2)$, the symmetry structure is richer, yielding seven distinct mobility sectors:
$\mathcal{M}\in\{\alpha,\beta_{x,1},\beta_{y,1},\beta_{y,2},\beta_{y,3},\beta_{y,6},\gamma\}$.
The lineon periods vary, and fusion leads to intricate period transmutation. When fusing $\beta_{y,2}$ and $\beta_{y,3}$ (lineons with periods $T=2$ and $T=3$), the outcome is unique:
\begin{equation}
    \beta_{y,2} \times \beta_{y,3} = \beta_{y,6},
\end{equation}
where the composite must satisfy the constraints of both input lineons. Since the periods $T=2$ and $T=3$ correspond to independent symmetry constraints, the composite inherits the least common multiple (LCM) of the periods, $T=6$. An example demonstrating this with $\mathfrak{m}_1=(1+x)(1+y+y^2)$ and $\mathfrak{m}_2=(1+x)(1+y^2)$ is shown in Fig.~\ref{fig:mobility_map}(d).

Conversely, fusing two $\beta_{y,6}$ lineons causes the lineon periods to split into multiple channels:
\begin{equation}\label{lineon6}
    \beta_{y,6} \times \beta_{y,6} = \alpha + \beta_{y,1} +\beta_{y,2}+ \beta_{y,3} + \beta_{y,6}.
\end{equation}
Unlike the previous case, here the constraints are generally not independent---both lineons share the same underlying symmetry structure. At certain relative positions and excitation configurations, partial cancellation of symmetry constraints can reduce the effective period. In general, two lineons along the same direction, with periods $T_1$ and $T_2$, can fuse into a lineon with period $T$, where $T$ divides $\operatorname{lcm}(T_1, T_2)$ (lineons along different directions always fuse into a fracton); the specific fusion channels depend on the detailed interference of symmetry charges at different relative positions. An example with $\mathfrak{m}_1=\mathfrak{m}_2=1+x$ (both belong to the $\beta_{y,6}$ channel) is shown in Fig.~\ref{fig:mobility_map}(e). The fusion process of these two excitations does not generate a $\beta_{y,2}$ channel, whose generation requires adjusting $\mathfrak{m}$. Another example with $\mathfrak{m}_1=1+x,~\mathfrak{m}_2=(1+x)(1+x+xy+xy^2)$ (both belong to the $\beta_{y,6}$ channel) generates $\beta_{y,2}$ and $\beta_{y,6}$ channels but not the others; its mobility map is shown in Fig.~\ref{fig:mobility_map}(f). These two examples cover all channels of Eq.~\eqref{lineon6}. The algebraic framework in the Supplemental Material analytically determines the mobility fusion channels of the three representative models studied here~\cite{sm}.

\textit{Discussion.---}
In this Letter, we uncovered a hidden multi-channel structure within subsystem symmetry-enriched topological phases. We demonstrated in three exactly solvable models that while the topological charges of these excitations remain Abelian, their mobility classes obey complex, multi-channel fusion rules. The decoration polynomial $f(x,y)$ translates the geometric constraints into a precise algebraic description. This opens two primary avenues for future research.

Theoretically, this polynomial description may extend to a broader range of stabilizer codes, including fracton models~\cite{SlagleKim2017,BoxiLi2026BF,Hu2026Holography}; perturbing such models away from their exactly solvable points would connect the present classification to the critical behavior of fracton topological order~\cite{Muhlhauser2020,Zhu2023prl,PhysRevResearch.4.033111,Wang2026KT}. The period-transmutation phenomena observed here suggest a connection to algebraic number theory that warrants investigation in a wider class of models.

On the quantum-information front, our polynomial language is related to that used for modern quantum low-density parity-check (qLDPC) codes~\cite{chen2025,FT3,Hu2025CodeStates}, whose physical realizations are now at the forefront of experimental quantum error correction~\cite{GoogleQuantumAI2023,Bluvstein2024}. Whether controlled mobility restrictions can inform error-syndrome management is an open question for future work; the present models are not proposed as practical codes. A separate question is whether subsystem-symmetry enrichment can help protect topological order against thermal fluctuations, as studied for self-correcting and finite-temperature topological memories~\cite{FT2,FT1,FT4,FT5}. Finally, probing these models via subdimensional entanglement entropy~\cite{LiYe2023Hierarchy,li2026subdimensionalentanglemententropygeometrictopological} will help disentangle the distinct physical contributions of topological order and subsystem symmetry to the restricted mobility.

\emph{Acknowledgments}---We thank Y.-A. Chen, M.-Y. Li, Y. Zhou, X.-Y. Huang, Y.-T. Hu, R.-J. Guo for helpful discussions.
This work was in part supported by the National Natural Science Foundation of China (NSFC) under Grants No. 12474149, Research Center for Magnetoelectric Physics of Guangdong Province under Grant No. 2024B0303390001, and Guangdong Provincial Key Laboratory of Magnetoelectric Physics and Devices under Grant No. 2022B1212010008.

\nocite{haah_commuting_2013, AtiyahMacdonald1969, ziegler1995, GAO2001501, ostrowski}

\appendix
\onecolumngrid
\makeatletter
\setcounter{secnumdepth}{3}
\renewcommand{\thesection}{\Roman{section}}
\renewcommand{\thesubsection}{\Alph{subsection}}
\renewcommand{\thesubsubsection}{\arabic{subsubsection}}
\renewcommand{\theparagraph}{\arabic{paragraph}}
\renewcommand{\thesubparagraph}{\arabic{subparagraph}}
\titleformat{\section}[block]
{\centering\normalfont\large\bfseries\MakeUppercase}
{\thesection.}
{1em}
{}
\titlespacing*{\section}
{0pt}{3.5ex plus 1ex minus .2ex}{2.3ex plus .2ex}
\titleformat{\subsection}[block]
{\normalfont\bfseries}
{\thesubsection.}
{1em}
{}
\titlespacing*{\subsection}
{0pt}{3.25ex plus 1ex minus .2ex}{1.5ex plus .2ex}
\titleformat{\subsubsection}[block]
{\normalfont\itshape}
{\thesubsubsection.}
{1em}
{}
\titlespacing*{\subsubsection}
{0pt}{3.25ex plus 1ex minus .2ex}{1.5ex plus .2ex}
\makeatother

\section{Algebraic Setup: Operator Formalism and Model Constructions}
\label{sec:algebraic-setup}

This Supplemental Material provides the detailed mathematical formalism supporting the results in the main text, including the operator construction and the algebraic analysis of mobility constraints and fusion channels for the three decoration polynomials studied in the Letter.

\subsection{Polynomial representation of Pauli operators}
\label{sec:polynomial-rep}

We consider a two-dimensional square lattice where each unit cell contains $Q=3$ qubits: one on the vertex and two on the edges (horizontal and vertical). Following Ref.\ \cite{haah_commuting_2013}, any translation-invariant Pauli operator can be encoded as a $2Q$-component vector of Laurent polynomials in variables $x$ and $y$ over the field $\mathbb{F}_2=\{0,1\}$. Explicitly,
\begin{equation}\label{eq:poly-rep}
    \mathcal{O}=
    \begin{pmatrix}
        \mathbf{f}_1(x,y) \\
        \mathbf{f}_2(x,y) \\
        \mathbf{f}_3(x,y) \\ \hline
        \mathbf{g}_1(x,y) \\
        \mathbf{g}_2(x,y) \\
        \mathbf{g}_3(x,y)
    \end{pmatrix},
\end{equation}
where the upper three components specify the locations of Pauli $X$ operators on the three qubit types (vertex, horizontal edge, vertical edge) and the lower three components specify the locations of Pauli $Z$ operators. In this representation, multiplication of operators corresponds to addition of their polynomial vectors (modulo $2$). The commutation relation between two operators $O_1, O_2$ is determined by the symplectic product
\begin{equation}
    \mathcal{O}_1 \cdot \mathcal{O}_2 \equiv \mathrm{tr}\big(\mathcal{O}_1^\dagger \lambda_Q \mathcal{O}_2\big) = 0 \;\Longleftrightarrow\; [O_1,O_2]=0,
\end{equation}
where $\lambda_Q$ is the standard $2Q\times 2Q$ symplectic matrix and $\mathcal{O}^\dagger(x,y) = \big(\,\overline{\mathbf{f}}(x,y)\;|\;\overline{\mathbf{g}}(x,y)\,\big)$ with $\overline{f}(x,y)=f(x^{-1},y^{-1})$. The trace $\mathrm{tr}$ extracts the constant term.

\subsection{Hamiltonian and the decoration polynomial}
\label{sec:Hamiltonian}

The three models introduced in the main text are defined by the Hamiltonian
\begin{equation}
    H = -\sum_v A_v - \sum_p B_p - \sum_v C_v,
\end{equation}
whose terms are specified through a decoration polynomial $f(x,y)\in \mathbb{F}_2[x^{\pm1},y^{\pm1}]$.
The vertex term $A_v$ and the plaquette term $B_p$ have the general form
\begin{equation}
    \mathcal{A}_v =
    \begin{pmatrix}
        0      \\ 1+\bar x \\ 1+\bar y\\ \hline
        f(x,y) \\0\\0
    \end{pmatrix},\qquad
    \mathcal{B}_p =
    \begin{pmatrix}
        0 \\0\\0\\ \hline
        0 \\ 1+y \\ 1+x
    \end{pmatrix}.
\end{equation}
$A_v$ contains Pauli $X$ on the four edges adjacent to $v$ and Pauli $Z$ on a set of nearby vertices determined by $f(x,y)$; $B_p$ is the standard toric-code plaquette operator.

We first introduce the following decomposition lemma, which is central to constructing the $C_v$ term:
\begin{lemma}\label{lemma1}
    For any $i_0, j_0 \in\mathbb{Z}$, define $P_0(x,y)=\sum_{k=0}^{i_0-1} x^k$ and $Q_0(x,y)=x^{i_0}\sum_{k=0}^{j_0-1} y^k$, with the convention $\sum_{k=0}^{n-1}:=-\sum_{k=n}^{-1}$ for $n<0$, which is legitimate in $\ff$. Then
    \begin{equation}
        (1+x)P_0 + (1+y)Q_0 = 1 + x^{i_0}y^{j_0}.
    \end{equation}
    Consequently, for a Laurent polynomial $f(x,y) \in \ff$ there exist $P,Q\in\ff$ satisfying $f = (1+x)P + (1+y)Q$ if and only if $f(1,1)=0$, i.e.\ if and only if $f$ has an even number of non-zero coefficients. All models in this work satisfy this condition.
\end{lemma}
\begin{proof}
    The first identity follows by direct expansion over $\mathbb{F}_2$:
    \begin{equation}
        (1+x)P_0 + (1+y)Q_0 = (1+x^{i_0}) + x^{i_0}(1+y^{j_0}) = 1 + x^{i_0}y^{j_0}.
    \end{equation}
    For the second claim, necessity is immediate since $(1+x)P+(1+y)Q$ vanishes at $x=y=1$. For sufficiency, $f(1,1)=0$ means $f$ has an even number of non-zero terms, so they may be paired arbitrarily. For a pair $x^{i_1}y^{j_1} + x^{i_2}y^{j_2}$, factor out $x^{i_1}y^{j_1}$ to obtain $x^{i_1}y^{j_1}(1 + x^{i_2-i_1}y^{j_2-j_1})$. Apply the first identity---now available for exponents of either sign---and multiply through by $x^{i_1}y^{j_1}$. Summing over all pairs yields $P = \sum P_0$ and $Q = \sum Q_0$.
\end{proof}

To construct a third vertex term $C_v$ that commutes with all $A_v$ and $B_p$, we use the decomposition
\begin{equation}\label{eq:f-decomp}
    f(x,y) = (1+x)P(x,y) + (1+y)Q(x,y),
\end{equation}
which by Lemma~\ref{lemma1} exists if and only if $f(1,1)=0$ (satisfied by all models below). We further require $\gcd(P,Q)$ to be a unit of $\ff$, where gcds are defined up to a monomial; this is necessary and sufficient for the syndromes achievable by symmetric string operators to be exactly $\id f$, and it is an extra condition on the decomposition rather than a consequence of Lemma~\ref{lemma1}.
With this decomposition, a valid $C_v$ term is
\begin{equation}
    \mathcal{C}_v =
    \begin{pmatrix}
        1 \\0\\0\\ \hline
        0 \\ \bar x P(\bar x,\bar y)\\ \bar y Q(\bar x,\bar y)
    \end{pmatrix}.
\end{equation}
One readily checks that $[A_v, C_{v'}]=0$ for all $v,v'$ and that $[B_p, C_v]=0$ by construction. All three stabilizer types thus commute.

\subsection{The three representative models}
\label{sec:three-models}

The decoration polynomials $f_1, f_2, f_3$ used in the main text, together with a choice of $P,Q$ satisfying \eqref{eq:f-decomp}, are:
\begin{equation}
    \begin{aligned}
        \text{Model (a):} & \quad f_1 = 1+y = (1+x)(1+y) + (1+y)\cdot x
        \;\Longrightarrow\; P=1+y,\; Q=x,                                                    \\[4pt]
        \text{Model (b):} & \quad f_2 = (1+x)(1+y) = (1+x)(1+y^2)+(1+y)(y+xy)
        \;\Longrightarrow\; P=1+y^2,\; Q=y+xy,                                               \\[4pt]
        \text{Model (c):} & \quad f_3 = (1+x)(1+y^2)(1+y+y^2)
        = (1+x)(1+y)+(1+y)(y^3+xy^3)                                                         \\
                          & \qquad\qquad\qquad\qquad\;\Longrightarrow\; P=1+y,\; Q=y^3+xy^3.
    \end{aligned}
\end{equation}
All Hamiltonian constructions used below are restricted to these three explicitly verified choices.
These choices guarantee that $\gcd(P,Q)$ is a unit in each case; they are selected by hand, since the pairing construction of Lemma~\ref{lemma1} does not in general deliver a coprime pair. Note also that $C_v$ is fixed by $(P,Q)$ and not by $f$ alone: two decompositions of the same $f$ differ by $(P,Q)\mapsto(P+(1+y)c,\,Q+(1+x)c)$, whence the corresponding $C_v$ terms differ by a plaquette stabilizer and generate the same stabilizer group.

The ground-state degeneracy on an $L\times L$ torus is $4$, independent of $f$. Indeed, there are $3L^2$ qubits and $3L^2$ stabilizer generators ($L^2$ each of $A_v$, $B_p$, $C_v$), but two global constraints $\prod_v A_v = \mathbbm{1}$ (which uses $f(1,1)=0$) and $\prod_p B_p = \mathbbm{1}$ reduce the independent generators to $3L^2-2$; solving the relation module shows that these are the only relations for every $L$, so the rank is exactly $3L^2-2$, giving
\begin{equation}
    \mathrm{GSD} = 2^{3L^2 - (3L^2-2)} = 4,
\end{equation}
confirming that the models all possess $\mathbb{Z}_2$ topological order.

\subsection{Symmetric string operators and the excitation map}
\label{sec:string-ops}

$m$-anyons are violations of the $B_p$ terms. Their creation and movement are realized by edge $X$ operators, which in the polynomial language correspond to vectors with non-zero entries only in the second and third upper components. However, generic edge $X$ strings also create $C_v$ excitations. To preserve the subsystem symmetries, an anyon string operator must commute with every subsystem symmetry generator. Since the generators are products of vertex $X$ operators, this constrains only the vertex-$Z$ component $\mathbf{g}_1$ of the operator, and by the adjointness relation $\mathrm{tr}(\overline{f d}\,\sigma)=\mathrm{tr}(\bar d\,\overline{f}\sigma)$ it amounts to $\mathbf{g}_1\in\id{f}$. Commuting with every $C_v$ term is necessary but not by itself a restriction: for any edge $X$ pattern one may choose $\mathbf{g}_1$ so as to cancel all $C_v$ violations, at the cost of breaking the symmetry.

The building blocks that commute with all $C_v$ are obtained from the symplectic product:
\begin{equation}\label{eq:sym-block}
    \mathcal{D}_v =
    \begin{pmatrix}
        0 \\ y Q(x,y)\\ x P(x,y)\\ \hline 0\\0\\0
    \end{pmatrix}.
\end{equation}
One verifies that $\mathcal{C}_v \cdot \mathcal{D}_{v'} = 0$ for any $v,v'$. For computing $m$-anyon mobility it suffices to consider the operators $d(x,y)\mathcal{D}_v$: a symmetric operator without vertex $X$ operators has $\mathbf{g}_1=cf$, and subtracting the particular solution $(\,c(1+\bar x),\,c(1+\bar y)\,)$, whose $B_p$ syndrome vanishes, leaves a homogeneous solution, which---because $\gcd(P,Q)$ is a unit---is a multiple of $\mathcal{D}_v$.

The excitation map of such an operator with respect to the $B_p$ stabilizer is
\begin{equation}
    \epsilon(d\mathcal{D}_v) = \big(0,\; S_B \cdot (d\mathcal{D}_v),\; 0 \big),
\end{equation}
where $S_B = \mathcal{B}_p^\dagger$ acts as the syndrome map for $B_p$, and we evaluate
\begin{equation}\label{eq:excitation-map}
    \begin{aligned}
        S_B \cdot (d\mathcal{D}_v)
         & = (0\;0\;0\,|\,0\;1+\bar y\;1+\bar x)\,
        \lambda_3\,
        \begin{pmatrix}
            0 \\ d\cdot y\cdot Q(x,y)\\ d\cdot x\cdot P(x,y)\\ \hline 0\\0\\0
        \end{pmatrix} \\[4pt]
         & = d(x,y)\big[(1+\bar y)\,y\,Q(x,y)+(1+\bar x)\,x\,P(x,y)\big] \\
         & = d(x,y)\big[(1+y)Q(x,y)+(1+x)P(x,y)\big]
        = d(x,y)\,f(x,y).
    \end{aligned}
\end{equation}
Thus the string operator creates $m$-anyon excitations precisely where the polynomial $d(x,y)$ fails to be a valid configuration of the pattern $f(x,y)$.

Consequently, a finite-support operator $d\mathcal{D}_v$ moves an $m$-anyon of type $\mathfrak{m}(x,y)$ (an element of the quotient module $M = \mathbb{F}_2[x^{\pm1},y^{\pm1}]/\langle f\rangle$) from one location to another if and only if there exists a translation $(a,b)$ such that
\begin{equation}\label{eq:string-condition}
    \epsilon(d\mathcal{D}_v) = (1 + x^a y^b)\,\mathfrak{m}(x,y).
\end{equation}
Equivalently, this means that $1+x^a y^b$ belongs to the annihilator ideal of $\mathfrak{m}$ in the ring $R=\mathbb{F}_2[x^{\pm1},y^{\pm1}]$ modulo $f$. This algebraic condition is the foundation for the mobility classification and fusion rules derived in the next sections.

\section{Mobility of excitations}
We first introduce the formal definition of the mobility polynomial, which algebraically encodes the set of positions an excitation can reach without breaking subsystem symmetries.

\begin{definition}
    \textit{Mobility Polynomial} $r(x,y)$ for a given excitation $\mathcal{E}$ is defined as the sum of monomials $x^iy^j$ such that $\mathcal{E}$ and $x^iy^j\mathcal{E}$ can be annihilated simultaneously by a symmetric string operator.

    For example, if some excitation $\E$ can move along the direction parallel to the $x$-axis, the mobility polynomial of this excitation writes
    \begin{equation}
        r_{\E}(x,y)=\sum_{i=-\infty}^{\infty} x^i.
    \end{equation}
\end{definition}

Given an anyon configuration $\mathfrak{m}(x,y)$ defined in the Letter, we have the following theorem to determine its mobility polynomial. Throughout we normalize $f\in\mathbb{F}_2[x,y]$ with $f(0,0)=1$, which is no loss of generality since a monomial factor merely translates the decoration pattern; this makes $\gcd(f,\mathfrak{m})$ and hence $g$ unique, so that the polynomial conditions below are well posed.

\begin{theorem}\label{theorem1}
    The mobility polynomial of excitation $\mathfrak{m} (x,y)$ is determined by a characteristic polynomial
    \begin{equation}
        g[f(x,y),\mathfrak{m}(x,y)]=\frac{f(x,y)}{\gcd(f(x,y),\mathfrak{m}(x,y))}\label{eq_thm}
    \end{equation}
    via the following three rules.
    First, if $g(x,y)=1$, then the excitation is fully mobile, with $r_{\mathcal{E}}(x,y)=\sum_{i,j\in\mathbb{Z}}x^iy^j$.
    Second, if $g(x,y)=t[q(x,y)]$ for the monomial $q(x,y)=x^uy^v$ with $(u,v)$ primitive, and polynomial $t(q)$ in $\mathbb{F}_2[q]$ with $t(0)=1$  (meaning that $t(q)$ is invertible (i.e., a unit) in $\mathbb{F}_2[[q]]$, the formal power series ring of $q$, containing elements like $\sum_{i=0}^\infty \lambda_i q^i $, where $\lambda_i\in \mathbb{F}_2$. See \cite{AtiyahMacdonald1969} Chapter 1 for more details). Let $t^{-1}(q)=\sum_{k=0}^\infty b_k q^k$ be the inverse of $t(q)$, then the excitation has linear mobility parallel to $(u,v)$, with $r_{\mathcal{E}}=\sum_{k=-\infty}^\infty q^{kT}$, where $T$ is the minimal positive integer such that $b_{k+T}=b_k$ for all $k\ge 0$; equivalently, the least $\ell\ge 1$ with $t(q)\mid 1+q^\ell$.
    Third, otherwise, the excitation is immobile with $r_{\mathcal{E}}(x,y)=1$.
\end{theorem}

To prove this theorem, we first introduce the following useful lemmas:

\begin{lemma}\label{lemma:symmetric}
    If $x^ay^b$ appears in $r(x,y)$, then so does $x^{-a}y^{-b}$. That is to say, the $r(x,y)$ is always symmetric under the antipode map $x\to x^{-1},y\to y^{-1}$.
\end{lemma}

\begin{proof} The proof follows directly from the translational symmetry (which we have assumed for our model). If there is a string operator creating two $\mathcal{E}$ excitations at $(0,0)$ and $(a,b)$, then the same string operator can create two $\mathcal{E}$ excitations at $(-a,-b)$ and $(0,0)$ being shifted by the vector $(-a,-b)$, adding the term $x^{-a}y^{-b}$ to the mobility polynomial.
\end{proof}
According to \textit{Lemma} \ref{lemma:symmetric}, we can consider only the $r(x,y)=\sum_{ij} r_{ij} x^i y^j$ with all $i\ge 0$. From now on we only consider mobility polynomials with positive powers along the direction of motion---the $x$ powers here, and the powers of $q$ in the linear case of Theorem~\ref{theorem1}---named \textit{positive mobility polynomial}. For the example above, we say its positive mobility polynomial $r_+(x,y)$ is
\begin{equation}
    r_{\E+}(x,y)=1+x+x^2+x^3+...
\end{equation}

\begin{lemma}\label{lemma:reversible}
    For $p(x,y),q(x,y)\in \mathbb{F}_2[x,y]$, if $p(0,0)=1$, then their characteristic polynomial $g(x,y)=g[p(x,y),q(x,y)]$ satisfies
    \begin{equation}
        g(0,0)=1.
    \end{equation}
\end{lemma}

\begin{proof}
    Let $\gcd(p,q)=d(x,y)$, and $p=dp'$, $q=dq' $. It follows that
    \begin{equation}
        p(0,0)=d(0,0)p'(0,0)=1,
    \end{equation}
    so $d(0,0)=1$. And we have
    \begin{equation}
        g(p,q)=\frac{p}{\gcd(p,q)}=\frac pd,
    \end{equation}
    so
    \begin{equation}
        g(0,0)=\frac{p(0,0)}{d(0,0)}=1.
    \end{equation}
\end{proof}

Now we prove Theorem \ref{theorem1}.
\begin{proof}
    We begin our proof by clarifying some concepts. Suppose the mobility polynomial of excitation $\mathcal E=\m(x,y)$ is \begin{equation}
        r_{\E}(x,y)=\sum_{i,j} r_{ij}x^iy^j,
    \end{equation}
    then it means that for all $r_{ij}=1$, there exists a symmetric, finitely supported string operator $\mc{O}_{ij}$, such that
    \begin{equation}
        \epsilon(\mc O_{ij})=(1+x^iy^j)\m (x,y).
    \end{equation}
    From Sec.~\ref{sec:string-ops} we learn that
    \begin{equation}
        \epsilon(\mc O_{ij})=d_{ij}(x,y)f(x,y)
    \end{equation}
    for some $d_{ij}(x,y)\in \mathbb F_2[x,y,\bar x,\bar y]$. That is to say, $\epsilon(\mc O_{ij})$ is in the ideal $\id{f}$ generated by $f(x,y)$, namely
    \begin{equation}
        \epsilon(\mc O_{ij})\in \id{f}.
    \end{equation}
    So we can finish our proof in two steps: first, finding all possible $h(x,y)$ such that $h(x,y)\m(x,y)\in\id f$; and second, checking whether $1+x^iy^j$ is one of the $h(x,y)$.
    By definition, the set of all possible $h(x,y)$ is exactly the quotient of two ideals:
    \begin{equation}
        h(x,y)\in \id f :\id \m.
    \end{equation}
    Since $\mathbb F_2[x,y,\bar x,\bar y]$ is a unique factorization domain (UFD) \cite{AtiyahMacdonald1969}, we have, directly from uniqueness of factorization,
    \begin{equation}
        \id f :\id \m=\id{\frac{f}{\gcd(f,\m)}},
    \end{equation}
    which contains all possible polynomial $h(x,y)$. Now we determine whether $1+x^iy^j$ is in $\id f:\id \m$. We start with some useful definitions and lemmas.

    \begin{definition}
        For non-zero $f\in\ff$, let $\supp(f)\subset\mathbb{Z}^2$ be the set of exponents of its non-zero terms and $\np(f):=\conv(\supp(f))$ its Newton polygon. In 2D, we say the dimension of a Newton's polygon $\dim\np(f)$ is $0$ if $\np(f)$ contains only 1 point, or $1$  if $\np(f)$ can be fully included in a straight line, or $2$ otherwise. It is unchanged when $f$ is multiplied by a monomial.
    \end{definition}

    \begin{definition}
        For two point sets $A,B$ in vector space, the Minkowski sum \cite{ziegler1995} $A+B$ is defined as
        \begin{equation}
            A+B:=\{a+b|a\in A,b\in B\}.
        \end{equation}
    \end{definition}

    \begin{lemma}[Ostrowski\cite{GAO2001501,ostrowski}]\label{lemma:ostrowski}
        For two polynomials $f,g\in\ff$, we have
        \begin{equation}
            \np (f) + \np(g)=\np(f\cdot g),
        \end{equation}
        where ``$+$'' is the Minkowski sum, and ``$\cdot$'' is the normal polynomial multiplication.
    \end{lemma}

    \begin{lemma}\label{lemma:greater}
        For $f,g\in\ff$,
        \begin{equation}
            \dim \np(f\cdot g)\ge \max (\dim\np(f),\dim \np(g)).
        \end{equation}
    \end{lemma}

    \begin{proof}
        Using Lemma \ref{lemma:ostrowski}, we have
        \begin{equation}
            \dim \np(f\cdot g)=\dim \left(\np(f)+\np(g)\right).
        \end{equation}
        WLOG, we assume $\dim \np(f)\ge\dim\np(g)$. Let point $a\in \np (g)$. Then
        \begin{equation}
            \dim(\np(f)+a) =\dim\np(f)
        \end{equation}
        since $\np(f)+a$ is nothing but a translation of $\np(f)$.
        And we have
        \begin{equation}
            \np(f)+a\subset \np(f)+\np(g)
        \end{equation}
        since $a\in\np(g)$. Therefore
        \begin{equation}
            \max (\dim\np(f),\dim \np(g))=\dim\np(f)=\dim(\np(f)+a)\le \dim(\np(f)+\np(g)).
        \end{equation}
    \end{proof}

    We now prove the theorem. Let
    \begin{equation}
        g[f(x,y),\m(x,y)]=\frac{f}{\gcd(f,\m)},
    \end{equation}
    then we have the following cases.
    First, if $g(x,y)=1$, or equivalently $\dim \np(g)=0$, then $\id g=\ff $, which is the original polynomial ring. Thus $\forall i,j\in\mb Z$, $1+x^iy^j\in\id g$. Thus the mobility polynomial is \begin{equation}
        r_{\E}(x,y)=\sum_{i,j\in\mb Z}x^iy^j,
    \end{equation}
    indicating that the excitation is fully mobile.

    Second, suppose $g(x,y)=t(x^uy^v)$, or equivalently, by Lemma~\ref{lemma:reversible}, $\dim \np(g)=1$, where $t$ is a one-variable polynomial. We determine the positive integers $\ell$ for which there exists a finitely supported polynomial $d_\ell(x,y)$ satisfying Eq.~\eqref{eq:string-condition},
    \begin{equation}
        d_\ell(q)\,t(q)=1+q^\ell, \qquad q:=x^uy^v.
    \end{equation}
    Only displacements parallel to $(u,v)$ can occur and $d_\ell$ is necessarily a polynomial in $q$ alone: grading $\ff$ by the cosets of the sublattice $\mathbb{Z}(u,v)$, multiplication by $t$ preserves each graded piece, so matching components forces both statements.
    We show that the allowed values are precisely $\ell=mT$ with $m\in\mathbb{Z}_+$, where $T$ is the minimal period defined below.

    To determine these integers, we use the inverse of $t$ only as an auxiliary formal power series. Since $t(0)=1$, $t$ is a unit in the formal-power-series ring $\mathbb{F}_2[[q]]$, and we set
    \begin{equation}
        d_{\infty}(q)\,t(q)=1.
    \end{equation}
    Equivalently,
    \begin{equation}
        d_{\infty}=t^{-1}(q).
    \end{equation}
    We obtain the allowed finitely supported $d_\ell$ from finite truncations of this series. In $\mathbb{F}_2[[q]]$, a polynomial has an inverse if and only if its constant term is $1$. Setting $N:=\deg t$, the polynomial $t(q)=\sum_{i=0}^N a_iq^i$ has $a_N=1$ automatically and $a_0=1$ by hypothesis, so it is invertible; $a_N=1$ is used below only to make the state transition invertible. We next show that its inverse $t^{-1}(q)=\sum_{i=0}^{\infty} b_iq^i$ is periodic: there exists $T>0$ such that $b_{i+T}=b_i$ for all $i\ge0$.

    Writing $t(q)t^{-1}(q)=1$ coefficient by coefficient gives
    \begin{equation}
        \begin{cases}
            a_0b_0                                & =1 \implies b_0=1                           \\
            a_1b_0+a_0b_1                         & =0 \implies b_1=a_1b_0                      \\
            a_2b_0+a_1b_1+a_0b_2                  & =0\implies b_2=a_2b_0+a_1b_1                \\

                                                  & \vdots                                      \\
            a_0b_k+a_1b_{k-1}+\cdots +a_N b_{k-N} & =0\implies b_k=a_1b_{k-1}+\cdots a_Nb_{k-N} \\
                                                  & \vdots
        \end{cases}
    \end{equation}

    Define $S_k=(b_{k-1},\ldots,b_{k-N})$, with $b_i=0$ for $i<0$. For $k\ge1$, the recurrence defines a linear transition $S_k\mapsto S_{k+1}$. Because $a_N=1$, this transition is invertible: the preceding last component is recovered uniquely from the new state. Hence the orbit of $S_1$ in the finite state space is periodic from its first state, rather than merely eventually periodic. Therefore there exists $T>0$ such that $b_{k+T}=b_k$ for all $k\ge0$; take $T$ minimal. It follows that $t^{-1}(q)$ has the form
    \begin{equation}
        \begin{aligned}
            t^{-1}[q(x,y)] & =\sum_{i=0}^{\infty} q^{iT}\left(\sum_{j=0}^{T-1} b_{j}q^{j}\right) \\
                           & = \sum_{i=0}^{\infty} q^{iT} p(q)                                   \\
                           & =\frac{p(q)}{1+q^T},
        \end{aligned}
    \end{equation}
    with defining $p(q)\equiv\sum_{j=0}^{T-1} b_{j}q^{j}$.

    Now we have reached the conclusion that
    \begin{equation}
        \begin{gathered}
            \forall t[q(x,y)] \text{ invertible, } \exists p[q(x,y)] \text{ and } T\in \mathbb{Z}_+ ,\\
            \text{such that } \frac{t(q)p(q)}{1+q^T}=1.
        \end{gathered}
    \end{equation}

    \begin{lemma}\label{lemma4}
        Now we claim that under truncation $\frac{p(q)}{1+q^T}\to \left[\frac{p(q)}{1+q^T}\right]_{[0,mT-1]}$ where $m\in \mathbb{Z}_+$,
        \begin{equation}
            \left[\frac{p(q)}{1+q^T}\right]_{[0,mT-1]}t(q)=1+q^{mT}.
        \end{equation}
    \end{lemma}

    \begin{proof}
        First notice that under this truncation, explicitly write out $p(q)=\sum_{i=0}^{T-1}p_iq^i$ and $p_0=1 $, we can write
        \begin{equation}\label{eq:period}
            \begin{aligned}
                t(q)\left[\frac{p(q)}{1+q^T}\right]_{[0,mT-1]} & =t(q)(1+q^T+q^{2T}+\cdots +q^{(m-1)T})(1+p_1q+\cdots +p_{T-1}q^{T-1}) \\
                                                               & =\left[\frac{1}{1+q^T}\right]_{[0,mT-1]}p(q)t(q)                      \\
                                                               & =(1+q^T+\cdots q^{(m-1)T})(1+q^T)                                     \\
                                                               & =1+q^{mT}.
            \end{aligned}
        \end{equation}
    \end{proof}

    For other possible truncations $\left[\frac{p(q)}{1+q^T}\right]_{[0,mT-1+c]}$ where $c\in\{1,2,\cdots,T-1\}$,
    \begin{equation}
        \begin{aligned}
            t(q)\left[\frac{p(q)}{1+q^T}\right]_{[0,mT-1+c]} & =t(q)\left[\frac{p(q)}{1+q^T}\right]_{[0,mT-1]}+t(q)\left[\frac{p(q)}{1+q^T}\right]_{[mT,mT-1+c]} \\
                                                             & =1+q^{mT}+q^{mT}t(q)(1+p_1q+...+p_{c-1}q^{c-1}).
        \end{aligned}
    \end{equation}

    \begin{lemma}\label{lemma5}
        We claim that $t(q)(1+p_1q+\cdots +p_{c-1}q^{c-1})\ne 1+q^p,~\forall p\in\mathbb{Z}_+\backslash \{T\}$.
    \end{lemma}

    \begin{proof}
        The proof is by contradiction. If $t(q)(1+p_1q+\cdots +p_{c-1}q^{c-1})=1+q^p$ and $p\ne T$, then because $T$ is the minimal positive period of the LFSR (i.e., the multiplicative order of the residue class of $q$ in the finite unit group $(\mathbb{F}_2[q]/(t(q)))^\times$), $t(q)$ dividing $1+q^p$ implies that $p$ must be a multiple of $T$. Hence,
        \begin{equation}\label{eq:cond1}
            \exists m> 1, \text{s.t. } p=mT.
        \end{equation}
        But on the other hand
        \begin{equation}
            \deg (t(q)p(q))=\deg (1+q^T)=T\ge \deg (t(q)[p(q)]_{[0,c-1]})=p,
        \end{equation}
        thus we can infer that
        \begin{equation}
            p\le T,
        \end{equation}
        contradicting with Eq.~(\ref{eq:cond1}).
    \end{proof}

    By Lemmas~\ref{lemma4} and~\ref{lemma5}, for any $k\in\mathbb{Z}_+$, if
    \begin{equation}
        \left[\frac{p(q)}{1+q^T}\right]_{[0,k]}t(q)=1+q^{p},
    \end{equation}
    then
    \begin{equation}
        p\in\{T,2T,3T,\cdots\},
    \end{equation}
    where $T$ is the minimal positive period of $t^{-1}(q)$. Conversely, every finitely supported solution arises this way, since $t\mid 1+q^\ell$ is equivalent to $q^\ell\equiv 1 \pmod{t}$ (note $t(0)=1$ makes $q$ invertible modulo $t$), so the set of admissible $\ell$ is closed under subtraction and equals $T\mathbb{Z}_+$; the allowed displacements are therefore precisely $\ell=mT$. By definition, the positive mobility polynomial $r(x,y)$ of this excitation $\mathcal{E}$ writes
    \begin{equation}
        r_{\E+}(q)=1+q^T+q^{2T}+q^{3T}+\cdots.
    \end{equation}
    According to Lemma \ref{lemma:symmetric}, the mobility polynomial is
    \begin{equation}
        r_\E=r_{\E+}(q)+r_{\E+}(\bar q)+1,
    \end{equation}
    which is exactly what we want, finishing the proof of this case.

    Finally, for any other cases, or $\dim \np(g)=2$, then $1+x^iy^j\notin \id g$ for all $(i,j)\ne(0,0)$. This follows directly from Lemma \ref{lemma:greater}. Since $\dim\np(1+x^iy^j)=1$ for $(i,j)\ne(0,0)$ and we are considering a 2D vector space, assume that $\exists h(x,y)\ne 0$ such that $hg=1+x^iy^j$, then we have
    \begin{equation}
        \dim\np(hg)\ge \max(\dim\np(h),\dim\np(g))=2,
    \end{equation}
    giving $1=\dim\np(1+x^iy^j)=\dim\np(hg)=2$, a contradiction. Therefore such $h(x,y)$ cannot exist. The excitation cannot be moved elsewhere without breaking the symmetry or creating other excitations, so
    \begin{equation}
        r_{\cal E}=1.
    \end{equation}

    This completes the proof of Theorem~\ref{theorem1}.
\end{proof}

\section{Fusion Rules of Mobility: Algebraic Derivation and Model Calculations}
\label{sec:mobility-fusion}

The following statement is algebraic for $f\in R$; in this work we apply it only to $f_1,f_2,f_3$, whose Hamiltonian realizations were verified above. It gives the precise mobility fusion rules for the classes defined in the Letter:
\begin{theorem}\label{theorem2}
    Let $R = \ff$. Let $f \in R \setminus \{0\}$ be a non-unit with the unique irreducible factorization $f = \prod_{j=1}^r p_j^{e_j}$, where $e_j \ge 1$, and $\{p_1, \dots, p_r\}$ is a set of pairwise non-associate irreducible elements in $R$. For any $\mathfrak{m} \in R$, define the characteristic polynomial of $\mathfrak{m}$ with respect to $f$ as
    \begin{equation}
        g(\mathfrak{m}) = \prod_{j=1}^r p_j^{k_j(\mathfrak{m})}, \quad \text{where} \quad k_j(\mathfrak{m}) = e_j - \min\big(v_{p_j}(\mathfrak{m}), e_j\big)\in\{0,1,\dots,e_j\},
    \end{equation}
    and $v_{p_j}$ denotes the $p_j$-adic valuation on $R$ (the exponent of $p_j$ in the unique factorization), with $v_{p_j}(0)=+\infty$. Since $v_{p_j}(hf)\ge e_j$ and monomials are units, the truncation at $e_j$ makes $g$ depend only on the class of $\mathfrak{m}$ in $M=R/\id f$ and satisfy $g(x^iy^j\mathfrak{m})=g(\mathfrak{m})$.

    Call a triple of integers $(k_{1,j}, k_{2,j}, k_j)$ in $\{0,1,\dots,e_j\}$ \emph{admissible at $j$} when: (i) $k_j = \max(k_{1,j}, k_{2,j})$ if $k_{1,j} \neq k_{2,j}$; (ii) $k_j = 0$ if $k_{1,j} = k_{2,j} = 0$; and (iii) $k_j \in \{0, 1, \dots, c_j\}$ if $k_{1,j} = k_{2,j} = c_j > 0$.

    \emph{Necessity.} For all $\mathfrak{m}_1, \mathfrak{m}_2 \in R$, writing $k_{1,j} = k_j(\mathfrak{m}_1)$, $k_{2,j} = k_j(\mathfrak{m}_2)$ and $k_j = k_j(\mathfrak{m}_1+\mathfrak{m}_2)$, the triple $(k_{1,j}, k_{2,j}, k_j)$ is admissible at $j$ for every $1 \le j \le r$.

    \emph{Sufficiency.} Conversely, given integers $k_{1,j}, k_{2,j}, k_j$ for $1 \le j \le r$ such that $(k_{1,j}, k_{2,j}, k_j)$ is admissible at $j$ for every $j$, there exist $\mathfrak{m}_1, \mathfrak{m}_2 \in R$ realizing them, i.e.\ with $g(\mathfrak{m}_1) = \prod_{j} p_j^{k_{1,j}}$, $g(\mathfrak{m}_2) = \prod_{j} p_j^{k_{2,j}}$ and $g(\mathfrak{m}_1 + \mathfrak{m}_2) = \prod_{j} p_j^{k_j}$.
\end{theorem}

\begin{proof}
    Let $v_j = v_{p_j}$ for $1 \le j \le r$. By definition, the exponents of $g(\mathfrak{m})$ are determined by the truncated valuations
    \begin{equation}
        \min\big(v_j(\mathfrak{m}), e_j\big) = e_j - k_j(\mathfrak{m}) \quad \text{for all } 1 \le j \le r,
    \end{equation}
    so throughout the proof we work with these truncated valuations rather than with $g$ itself.

    \textbf{Necessity:} Fix an index $j$. We determine the possible values of $k_j = e_j - \min\big(v_j(\mathfrak{m}_1 + \mathfrak{m}_2), e_j\big)$ given $k_{1,j}$ and $k_{2,j}$.
    When $k_{1,j}\ne k_{2,j}$, assume without loss of generality that $k_{1,j}>k_{2,j}$. Then $k_{1,j}>0$, so $v_j(\mathfrak{m}_1)=e_j-k_{1,j}<e_j$. Moreover, $\min(v_j(\mathfrak{m}_2),e_j)=e_j-k_{2,j}>e_j-k_{1,j}$, which implies $v_j(\mathfrak{m}_2)>v_j(\mathfrak{m}_1)$ even when $k_{2,j}=0$. Hence $v_j(\mathfrak{m}_1+\mathfrak{m}_2)=v_j(\mathfrak{m}_1)=e_j-k_{1,j}$, and therefore $k_j=k_{1,j}=\max(k_{1,j},k_{2,j})$.
    When $k_{1,j} = k_{2,j} = 0$, we have $v_j(\mathfrak{m}_1) \ge e_j$ and $v_j(\mathfrak{m}_2) \ge e_j$. By the ultrametric inequality, $v_j(\mathfrak{m}_1 + \mathfrak{m}_2) \ge e_j$. Consequently, $\min\big(v_j(\mathfrak{m}_1 + \mathfrak{m}_2), e_j\big) = e_j$, yielding $k_j = 0$.
    When $k_{1,j} = k_{2,j} = c_j > 0$, we find $v_j(\mathfrak{m}_1) = v_j(\mathfrak{m}_2) = e_j - c_j < e_j$. It follows that $v_j(\mathfrak{m}_1 + \mathfrak{m}_2) \ge e_j - c_j$. Applying the minimum operation with $e_j$, we obtain $e_j - c_j \le \min\big(v_j(\mathfrak{m}_1 + \mathfrak{m}_2), e_j\big) \le e_j$. Subtracting this from $e_j$ yields $0 \le k_j \le c_j$. Thus $k_j \in \{0, 1, \dots, c_j\}$.

    \textbf{Sufficiency:} Let $k_{1,j}, k_{2,j}, k_j$ $(1 \le j \le r)$ be integers such that each triple $(k_{1,j}, k_{2,j}, k_j)$ is admissible at $j$.
    First, we specify local target elements $w_{1,j}, w_{2,j} \in R$ for each $j$ such that:
    \begin{equation}\label{eq:local_target}
        \min\big(v_j(w_{i,j}), e_j\big) = e_j - k_{i,j} \quad (i=1,2), \quad \text{and} \quad \min\big(v_j(w_{1,j} + w_{2,j}), e_j\big) = e_j - k_j.
    \end{equation}
    We use below that $R/(p_j)$ has more than two elements: since $R = \ff$ is a domain of Krull dimension 2 and $(p_j)$ has height 1, $R/(p_j)$ is a one-dimensional domain over $\mathbb{F}_2$, hence infinite.

    If $k_{1,j} \neq k_{2,j}$, we set $w_{1,j} = p_j^{e_j - k_{1,j}}$ and $w_{2,j} = p_j^{e_j - k_{2,j}}$.
    If $k_{1,j} = k_{2,j} = 0$, we set $w_{1,j} = p_j^{e_j}$ and $w_{2,j} = p_j^{e_j}$; in characteristic 2 their sum vanishes, and the convention $v_j(0) = +\infty$ gives $\min\big(v_j(w_{1,j}+w_{2,j}), e_j\big) = e_j$, as required by $k_j = 0$.
    If $k_{1,j} = k_{2,j} = c_j > 0$ and $k_j < c_j$, let $\ell = c_j - k_j > 0$. We then set $w_{1,j} = p_j^{e_j - c_j}$ and $w_{2,j} = p_j^{e_j - c_j}(1 + p_j^\ell)$. In characteristic 2, $w_{1,j} + w_{2,j} = p_j^{e_j - k_j}$.
    If $k_{1,j} = k_{2,j} = c_j > 0$ and $k_j = c_j$, we choose an arbitrary element $\bar{u}_1 \in R/(p_j) \setminus \{\bar{0}, \bar{1}\}$, which exists by the remark above, and let $u_1 \in R$ be a lift. Set $w_{1,j} = p_j^{e_j - c_j} u_1$ and $w_{2,j} = p_j^{e_j - c_j}$. Since $u_1 \not\equiv 0$ and $u_1 + 1 \not\equiv 0 \pmod{p_j}$, the valuations of $w_{1,j}, w_{2,j}$ and their sum are all strictly $e_j - c_j$.

    In all cases, \eqref{eq:local_target} is satisfied.

    To elevate these local elements to global elements in $R$, we define the multiplicatively closed set $S = R \setminus \bigcup_{j=1}^r (p_j)$.
    By the Prime Avoidance Lemma, any prime ideal of $R$ contained in $\bigcup_{j=1}^r (p_j)$ must be contained in some $(p_j)$. Consequently, the only non-zero prime ideals in the localization $R_S$ are $(p_1)R_S, \dots, (p_r)R_S$, which are maximal. Since $R_S$ is a regular domain of Krull dimension 1 with finitely many maximal ideals, $R_S$ is a semilocal principal ideal domain.
    Within this PID $R_S$, the elements $p_j$ remain pairwise non-associate, meaning the ideals $(p_j^{e_j+1})R_S$ are strictly pairwise comaximal. By the Chinese Remainder Theorem in $R_S$, there exist elements $M_1, M_2 \in R_S$ satisfying the simultaneous congruences:
    \begin{equation}
        M_i \equiv w_{i,j} \pmod{p_j^{e_j+1} R_S} \quad \text{for all } 1 \le j \le r \text{ and } i \in \{1, 2\}.
    \end{equation}
    Since $M_1, M_2 \in R_S$, there exists a common denominator $s \in S$ such that $s M_1 \in R$ and $s M_2 \in R$. Define global elements $\mathfrak{m}_i = s M_i$. By the definition of $S$, $s \notin (p_j)$ for all $j$, which enforces $v_j(s) = 0$. Hence $v_j(\mathfrak{m}_i) = v_j(M_i)$, and $v_j(\mathfrak{m}_1 + \mathfrak{m}_2) = v_j(M_1 + M_2)$.

    Finally, we verify the truncated valuations. From $M_i - w_{i,j} \in p_j^{e_j+1}R_S$, we have $v_j(M_i - w_{i,j}) \ge e_j + 1$.
    The non-Archimedean minimum properties yield:
    \begin{equation}
        \min\big(v_j(\mathfrak{m}_i), e_j\big) = \min\big(v_j(M_i), e_j\big) = \min\big(v_j(w_{i,j}), e_j\big) = e_j - k_{i,j}.
    \end{equation}
    Similarly, $v_j\big((M_1 + M_2) - (w_{1,j} + w_{2,j})\big) \ge e_j + 1$. Truncating at $e_j$, we obtain:
    \begin{equation}
        \min\big(v_j(\mathfrak{m}_1 + \mathfrak{m}_2), e_j\big) = \min\big(v_j(M_1 + M_2), e_j\big) = \min\big(v_j(w_{1,j} + w_{2,j}), e_j\big) = e_j - k_j.
    \end{equation}
    Because the evaluations match identically at each irreducible factor up to $e_j$, the global elements satisfy $g(\mathfrak{m}_1) = \prod p_j^{k_{1,j}}$, $g(\mathfrak{m}_2) = \prod p_j^{k_{2,j}}$, and $g(\mathfrak{m}_1 + \mathfrak{m}_2) = \prod p_j^{k_j}$.
\end{proof}

Applied to $f_1,f_2,f_3$, this theorem gives the mobility fusion rules reported in the Letter. Since distinct exponent vectors may represent the same mobility class, the class-level rule is the union of the outcomes of Theorem~\ref{theorem2} over all representatives of each input class; we note below where a representative is unique, so that a single one suffices.

\subsection{Model (a): Fibonacci-like fusion}
\label{sec:model-a-calc}

For Model (a), the decoration polynomial is $f_1(x,y) = 1+y$. In $\mathbb{F}_2[x,y]$, $f_1$ has a single irreducible factor $p_1 = 1+y$ with multiplicity $e_1 = 1$. The characteristic polynomial of any excitation $\mathfrak{m}$ takes the form $g(\mathfrak{m}) = (1+y)^{k_1}$, where $k_1 \in \{0, 1\}$.

The mobility class is uniquely determined by this exponent. When $k_1 = 0$, $g(\mathfrak{m}) = 1$, and Theorem~\ref{theorem1} dictates that the excitation is fully mobile ($\alpha$). When $k_1 = 1$, $g(\mathfrak{m}) = 1+y$. This polynomial has a one-dimensional Newton polygon, and the formal inverse of $t(y) = 1+y$ is $1/(1+y) = \sum_{n=0}^{\infty} y^n$, which has a period $T=1$. This corresponds to a lineon moving along the $y$-axis ($\beta_{y,1}$).

Considering the fusion rule $\beta_{y,1} \times \beta_{y,1}$, both input excitations carry the exponent $k_{1,1} = k_{2,1} = 1$. According to Theorem~\ref{theorem2}, the characteristic polynomial of their composite $g(\mathfrak{m}_1+\mathfrak{m}_2) = (1+y)^{k_1}$ must have an exponent satisfying $k_1 \in \{0, 1\}$. These two allowed values exactly correspond to the mobility classes $\alpha$ and $\beta_{y,1}$. Since the theorem guarantees both channels can be physically realized for appropriate local configurations, we directly obtain the Fibonacci-like fusion rule $\beta_{y,1} \times \beta_{y,1} = \alpha + \beta_{y,1}$.

\subsection{Model (b): Tensor product algebra}
\label{sec:model-b-calc}

For Model (b), the decoration polynomial $f_2(x,y) = (1+x)(1+y)$ factors into two distinct irreducible polynomials: $p_1 = 1+x$ and $p_2 = 1+y$, each with multiplicity $e_1=e_2=1$. The characteristic polynomial of an excitation takes the form $g(\mathfrak{m}) = (1+x)^{k_1}(1+y)^{k_2}$, labeled by the exponent vector $(k_1, k_2) \in \{0, 1\}^2$.

The four possible exponent vectors map directly to four distinct mobility classes. The trivial vector $(0,0)$ gives $g(\mathfrak{m}) = 1$, representing the fully mobile class $\alpha$. The vectors $(1,0)$ and $(0,1)$ yield $1+x$ and $1+y$, which are horizontal ($\beta_{x,1}$) and vertical ($\beta_{y,1}$) lineons, respectively. Finally, $(1,1)$ gives $g(\mathfrak{m}) = (1+x)(1+y)$. Because its Newton polygon has dimension 2, this excitation is an immobile fracton ($\gamma$).

Since $p_1$ and $p_2$ are distinct irreducible factors, Theorem~\ref{theorem2} implies that the exponent fusion for $k_1$ and $k_2$ proceeds entirely independently. The fusion rule for each component is simply $1 \times 1 \to \{0, 1\}$. The overall mobility fusion rules therefore have the product-form binary expressions obtained from two independent Fibonacci rules.

As a concrete example, fusing two fractons $\gamma \times \gamma$ means fusing $(1,1)$ with $(1,1)$. Applying the theorem component-wise restricts the composite exponents to $k_1 \in \{0, 1\}$ and $k_2 \in \{0, 1\}$. This generates all four possible exponent combinations, reproducing the rule $\gamma \times \gamma = \alpha + \beta_{x,1} + \beta_{y,1} + \gamma$. Similarly, the fusion of two orthogonal lineons $\beta_{x,1} \times \beta_{y,1}$ corresponds to the inputs $(1,0)$ and $(0,1)$. The non-matching exponents must take their maximum values, resulting in a unique composite vector $(1,1)$. This confirms the deterministic fusion channel $\beta_{x,1} \times \beta_{y,1} = \gamma$.

\subsection{Model (c): Period transmutation and splitting}
\label{sec:model-c-calc}

For Model (c), the decoration polynomial is $f_3(x,y) = (1+x)(1+y^2)(1+y+y^2)$. Over $\mathbb{F}_2$, it factors as $(1+x)^1 (1+y)^2 (1+y+y^2)^1$. We identify three irreducible factors: $p_1 = 1+x$ ($e_1=1$), $p_2 = 1+y$ ($e_2=2$), and $p_3 = 1+y+y^2$ ($e_3=1$). The characteristic polynomial $g(\mathfrak{m}) = p_1^{k_1} p_2^{k_2} p_3^{k_3}$ is labeled by the vector $(k_1, k_2, k_3) \in \{0,1\} \times \{0,1,2\} \times \{0,1\}$.

The six vectors with $k_1=1$ are immediate: $(1,0,0)$ gives $g=1+x$, a horizontal lineon $\beta_{x,1}$, while $k_1=1$ with $(k_2,k_3)\ne(0,0)$ gives $\dim\np(g)=2$ by Lemma~\ref{lemma:ostrowski}, hence a fracton $\gamma$; this accounts for the remaining two of the seven sectors quoted in the Letter. To analyze the rich structure of vertical lineons, we restrict our attention to $k_1 = 0$, giving $g(\mathfrak{m}) = (1+y)^{k_2} (1+y+y^2)^{k_3}$; this is without loss of generality, since $k_1=0$ is preserved under fusion and none of the classes below admits a representative with $k_1=1$. The period of each mobility class is extracted from the formal inverse of this polynomial. The trivial vector $(0,0)$ gives $g=1$, corresponding to $\alpha$. For $k_3=0$, the vectors $(1,0)$ and $(2,0)$ yield $1+y$ and $1+y^2$, whose inverses have periods $T=1$ ($\beta_{y,1}$) and $T=2$ ($\beta_{y,2}$), respectively. When $k_3=1$, the vector $(0,1)$ yields $1+y+y^2$, which has a period of $T=3$ ($\beta_{y,3}$). The composite vector $(1,1)$ gives $(1+y)(1+y+y^2) = 1+y^3$, which also has a period of $T=3$ ($\beta_{y,3}$). Finally, $(2,1)$ yields $(1+y)^2(1+y+y^2) = 1+y+y^3+y^4$, whose formal inverse exhibits a period of $T=6$ ($\beta_{y,6}$).

This algebraic mapping naturally explains the convergence and splitting behaviors of lineon periods. For the period convergence process $\beta_{y,2} \times \beta_{y,3} = \beta_{y,6}$, the inputs are $(2,0)$ and $(0,1)$; the class $\beta_{y,3}$ is also represented by $(1,1)$, which gives the same outcome, so the union is a single channel. Since the exponents differ for each irreducible factor, Theorem~\ref{theorem2} requires the composite exponents to be the component-wise maxima: $k_2 = \max(2,0) = 2$ and $k_3 = \max(0,1) = 1$. This uniquely dictates the vector $(2,1)$, which is exactly $\beta_{y,6}$. The periods combine via the least common multiple because they originate from distinct irreducible constraints.

In contrast, the period splitting process occurs when fusing identical constraints, such as $\beta_{y,6} \times \beta_{y,6}$. The inputs are both $(2,1)$. Theorem~\ref{theorem2} now permits the composite exponents to take any value up to the input values: $k_2 \in \{0,1,2\}$ and $k_3 \in \{0,1\}$. These 6 possible combinations map to the physical mobility classes as follows: $(0,0)$ maps to $\alpha$, $(1,0)$ to $\beta_{y,1}$, $(2,0)$ to $\beta_{y,2}$, both $(0,1)$ and $(1,1)$ map to $\beta_{y,3}$, and $(2,1)$ maps to $\beta_{y,6}$. Grouping these classes together precisely yields the multi-channel splitting rule $\beta_{y,6} \times \beta_{y,6} = \alpha + \beta_{y,1} + \beta_{y,2} + \beta_{y,3} + \beta_{y,6}$. The complex transmutation of lineon periods is thus a rigorous consequence of the underlying $p$-adic valuation structure.

\end{document}